\documentclass[runningheads]{llncs}

\usepackage{amsmath,amssymb,amstext} %
\usepackage[pdftex]{graphicx}
\usepackage[citestyle=numeric,bibstyle=trad-plain,sorting=nty, doi=false,isbn=false]{biblatex}

\usepackage[utf8]{inputenc}
\usepackage{siunitx}
\usepackage{graphicx}
\usepackage[dvipsnames]{xcolor}
\usepackage{tikz}
\usepackage{comment}
\usepackage{pgfplots}
\usetikzlibrary{patterns}
\usepackage{url}
\usepackage{pdfpages}
\usepackage{subcaption}
\usetikzlibrary{shapes, arrows, arrows.meta, calc, positioning, fit,backgrounds}
\pgfplotsset{compat=newest}
\usepackage{multirow}

\usepackage{tabularx}
\usepackage{arydshln}
\usepackage{enumerate}
\usepackage{mathtools}
\usepackage{amsmath}
\usepackage{multicol}
\usepackage{algorithm}
\usepackage{algpseudocode}
\usepackage{xspace}
\newcommand{\scheme}{MSPT\xspace}

\usepgfplotslibrary{colorbrewer}

\usepackage[pdftex,pagebackref=false]{hyperref} %

	\title{Multi-Shard Private Transactions for Permissioned Blockchains}
	\author{
	      Elli Androulaki\inst{1} \and
              Angelo De Caro\inst{1} \and
              Kaoutar Elkhiyaoui\inst{1}\and
              Christian Gorenflo\inst{2}\and
              Alessandro Sorniotti\inst{1}\and
              Marko Vukolic\inst{1}
              }
	\institute{
	IBM Research - Zurich\\ \and
	University of Waterloo
	}

\bibliography{Mendeley}
\bibliography{blockchain}
\bibliography{fault_tolerance}
\bibliography{local}
\bibliography{related}

\begin{document}
\maketitle

	\renewcommand{\arraystretch}{1.5}

	\begin{abstract}

	Traditionally, blockchain systems involve sharing transaction information across all blockchain network participants.  Clearly, this introduces barriers to the adoption of the technology by the enterprise world, where preserving the privacy of the business data is a necessity. 
	Previous efforts to bring privacy and blockchains together either still leak partial information, are restricted in their functionality~\cite{Quorum2019} or use costly mechanisms like zk-SNARKs~\cite{BCGGMT14}. In this paper, we propose the Multi-Shard Private Transaction (MSPT) protocol, a novel privacy-preserving protocol for permissioned blockchains, which relies only on simple cryptographic primitives and targeted dissemination of information to achieve atomicity and high performances. 
	\end{abstract}
		
\section{Introduction} 
	Early blockchain platforms have enabled decentralized payments that forego classical financial institutions altogether. Trust is established through full decentralization and consensus mechanisms (e.g., Proof-of-Work and Proof-of-Stake)~\cite{Nakamoto2008,pos}. However, pain points of many of the early platforms have been privacy and scalability. While a number of privacy solutions based on zero-knowledge proofs have been proposed to mitigate privacy threats~\cite{MGGR13,BCGGMT14,BAZB19,FMMO18}, these do not perform or scale well. Meanwhile, in particular in the blockchain for business context \cite{HyperledgerURL}, a shift towards Byzantine fault-tolerant (BFT)-based permissioned blockchains is under way to accommodate throughput-hungry enterprise applications. 

	This paper aims at bringing together privacy and scalability to \emph{permissioned blockchains}, without taxing transactional throughput; enabling thus a wider adoptions by the enterprise world. Our starting point is that enterprises require the decentralization of the blockchain alongside properties such as discoverability, single source of truth and distributed trust. However they do also require certain transactions to be siloed from the rest of the network owing to security and privacy concerns. In short, it is required that the blockchain behave both as a decentralized platform \textit{and} as a sparse collection of information silos when it comes to private n-lateral transactions within a larger network. While a combination of zero-knowledge proofs and encryption would allow a privacy-preserving migration to the blockchain, it would incur a considerable computational overhead to produce and verify transactions -- not to mention loss of functionality. Instead, we let private data reside in their natural silos and use the blockchain to facilitate transactions that impact multiple silos. Multiple information silos spontaneously arise in various industries, from the content, context and participants of any private legal agreement, through the details of any insurance contract to financial trusts. A silo in this context perfectly matches a blockchain \emph{shard}. Sharding is a common technique~\cite{Hearn2019, Damgard2020} to partition the blockchain, in order to increase transactional throughput and alleviate storage requirements~\cite{shardinc}. Yet, as of now, there are no privacy-preserving cross-shard transaction processing mechanisms that scale: validating such transactions require the blockchain to either access the content of the transactions~\cite{Amiri2019b} or the use of advanced cryptographic primitives (e.g., zero-knowledge proofs). 
	
	In this paper, we introduce \textit{Multi-shard private transactions} \scheme to bridge this gap. \scheme is a \emph{privacy-preserving atomic commit protocol} for permissioned blockchains, which allows shards to update their states atomically while keeping the computational overhead and information disclosure at the blockchain minimal. 
	\smallskip

	\noindent{\bf Contributions.} The contributions of our paper are three-fold:
	\begin{itemize}
		\item We introduce \scheme: a privacy-preserving protocol that allows shards to atomically update their states following a blockchain transaction. Notably, the protocol ensures that:  (i) at the end of its execution, shards either commit or discard the transaction unanimously;  (ii) the transaction only reveals the number of involved shards to the blockchain; (iii) finally,  during the protocol execution, a shard only learns the identities of the other shards on which its state update depends. 
		While we focus our description on a simplified MSPT version where a shard is operated by a single node, we also offer a generalized version of MSPT that distributes the shard operation across crash-prone nodes.
		
		\item A privacy-preserving atomic commit (PPAC) protocol, which sits at the heart of \scheme. PPAC can be seen as generalization of the celebrated two phase commit (2PC) protocol and we believe may be of independent interest.
		\item A proof-of-concept implementation and an experimental analysis of  simplified \scheme based on Hyperledger Fabric v1.4.
	\end{itemize}

	\smallskip

	\noindent\textbf{Related Work.} 
	Hyperledger Fabric~\cite{Fabric2019} leverages \emph{channels} to completely separate private from public data. Unfortunately, since Fabric does not implement cross-channel communication, the usefulness of such an architecture is limited. Androulaki~et~al.~\cite{androulaki2018channels} explored cross-channel transactions to move an asset from one channel to another. However, the proposed solution requires participants to go through three non-atomic steps: locking the asset on its channel of origin, exchanging the cross-channel transaction to prove the validity of the asset and finally unlocking the asset on its new channel. 
	
	Hyperledger Fabric comes also with \emph{private data collections} that help participants on the same channel restrict access to their private data. In a nutshell, a transaction will not carry the new value of the state being updated, rather only the hash is transmitted. This approach, though, reveals which parties are involved in a transaction and leaks access patterns 
	Our \scheme protocol, on the other hand,  completely prevents network participants from learning any information they are not privy to.

	CAPER~\cite{Amiri2019b} leverages sharding to increase data privacy and to scale performance. Each shard maintains an internal private chain of transactions that is intertwined with a global public chain. A shard orders and executes its  internal transactions locally, however, cross-shard transactions take place on the public chain and everyone can access their content.

	Quorum, an enterprise blockchain based on Ethereum, substitutes transaction payload with hashes and only grants authorized parties access to the payload~\cite{Quorum2019}. However, it is not possible to create transactions that span multiple trust boundaries. Malicious nodes can stall the system or create inconsistent states; for instance, by submitting a transaction hash to the ledger and timing the dissemination of the payload, so that only some nodes receive it.

	Atomic commit protocols (ACP) ensure that a distributed system either collectively commits or aborts an operation~\cite{Chrysanthis1998}. %
	Existing work focuses on either efficiency assuming only honest participants, or fault-tolerance in the presence of failures. 
	In particular, \emph{non-blocking} ACPs allow correct nodes to make progress without waiting for any failed nodes~\cite{keidar1995increasing, guerraoui2002non, Guerraoui2002}. Our scheme however, aims at efficiency, and hence, uses \emph{blocking} ACPs to achieve atomicity. The novelty of our scheme though lies in preserving privacy during the protocol execution. To the best of our knowledge, we are the first to investigate the angle of privacy in the context of atomic commit protocols.

	\smallskip
	
	\noindent\textbf{Roadmap.} The rest of this paper is organized as follows. In Section~\ref{sec:prelim} we introduce system model and terminology. In Section~\ref{sec:black-box} we introduce a simplified version of MSPT where shards are centralized. Section~\ref{sec:PPC} focuses on the privacy-preserving atomic commit (PPAC) subprotocol of MSPT. We evaluate MSPT in Section~\ref{sec:tngpoc} and discuss the extension of MSPT beyond Fabric and to distributed shards in Section~\ref{sec:generalize}.

\section{Preliminaries}
	\label{sec:prelim}
	Before delving into technical details,  we provide an overview of the model, including system participants, trust and network assumptions (Sec.~\ref{sec:model}), as well as an overview of \scheme (Sec.~\ref{subsec:protocol-overview}) and its goals (Sec~\ref{sec:spec}).

    \subsection{System Model}
	\label{sec:model}
	
	 \noindent\textbf{State.}
    State is very broadly defined as any data which is kept by the blockchain system. Its integrity and correctness must be guaranteed by the system since it is assumed that changes over it have an impact in the real world (for instance state data may correspond to ownership over financial assets, deeds to physical objects, entitlements etc.). As a result, the system enforces that data be modifiable only by a selected set of system actors, which we refer to as \textit{stakeholders}.
		
	 \noindent\textbf{Stakeholders.}
	We loosely define a state stakeholder (henceforward referred to as \emph{stakeholder}) as any party that stands to gain or lose something from the modification of a specific state in the system, i.e. they have a stake in it. Access rules follow naturally: a state can be read only by its stakeholders and can be modified only if a predefined policy, called hereafter \emph{stakeholder policy}, is adhered to, e.g. \emph{all of its stakeholders agree} or \emph{n out of its m stakeholders agree}. State is modified when stakeholders assemble and submit \textit{transactions}. %
	Stakeholders satisfying a stakeholder policy are honest w.r.t. approving updates to, and preserving the privacy of, that state, but stakeholders may otherwise form arbitrary coalitions to modify or access state they are not authorized to modify or access.

    \noindent\textbf{Shards.} A shard is the logical entity inside the blockchain system that 
    guarantees integrity and correctness of 
    the states the stakeholders have give the shard control of.
    A shard consists of a non-empty set of \emph{shard nodes}. Shard nodes share the same view of a subset of the system state by maintaining a private state database including state data of many stakeholders. As previously described, the access and ability to modify those states is governed by policies on a state by state basis. Shard nodes belong to the same trust boundary (for example, the same organization) w.r.t. the data that the shard handles. A shard as a whole is assumed to be \emph{honest-but-curious}, yet a fraction $f$ of the nodes in the shard can be crash faulty. Although we assume non-malicious behavior, we will outline the protocol's defenses and limitations in a Byzantine setting. In addition to the private state database, shards also replicate the content of the \emph{ledger}.
	
	\noindent\textbf{Ledger.}
	The ledger is a decentralized data store that maintains an ordered sequence of all \textit{transactions} of the permissioned blockchain. It is implemented over \emph{ledger nodes} (potentially overlapping with shard nodes), that keep \emph{ordered} records of all \emph{valid} transactions ever submitted. A transaction is deemed valid if it has been submitted by an enrolled member of the blockchain. 
	The ledger is assumed to be available at all times and shards rely on the ledger to order transactions. %
	Ledger nodes may be crash or Byzantine faulty, where the abstraction of a highly available ledger is implemented using a state-of-the-art crash fault-tolerant (CFT) or Byzantine fault-tolerant (BFT) state machine replication protocol (e.g., \cite{Stathakopoulou2019}). %
    
	\noindent\textbf{Clients.} These are enrolled members of the blockchain system that can submit transactions on behalf of stakeholders. They can be malicious in the same way the stakeholders can be.
    		
	\noindent\textbf{Network.}  We assume that the communication between all participants is point-to-point authenticated (e.g., using TLS~\cite{rescorla2018}) and that the communication in the network is partially synchronous \cite{Dwork1988}. More specifically, there is an upper bound on latency $\Delta<\infty$ for every message sent between two honest participants. Consequently, liveness of communication between any number of non-faulty participants is guaranteed. 
	Additionally, we consider that network-level adversaries are out of scope. 
    
	\noindent\textbf{Identity management.}
	All entities in the system (stakeholders, shard nodes, clients...) must obtain certified public keys with which they can authenticate to other entities and request authorization to perform certain actions. We refer to these are \textit{long-term public keys}. Long-term public  keys are used by their possessor for a prolonged time frame (e.g. years) and they are tied to the true identity of the possessor. At the same time, an actor may want to interact with the blockchain in a way that better preserves their privacy, e.g. by not revealing their true identity but only their entitlement to perform an action. A second set of public keys called \textit{ephemeral public keys} are used in these cases.

	\subsection{\scheme Goals} 
	\label{sec:spec}
	
	\noindent{\bf Correctness.} Correctness is the property that guarantees that state will only change when its stakeholders authorize it, and only in the way prescribed by them. We assume that an honest-but-curious shard will not update a state unless the corresponding stakeholder policy is met and the update satisfies the shard's internal validation rules, even in the presence of malicious shards, stakeholders and clients. Malicious shards will not obey these restrictions.
	
	Under these assumptions, \scheme ensures that if the protocol execution only involves honest-but-curious shards, then correctness is guaranteed. Notably, shards will either unanimously commit or unanimously discard the resulting state changes. No assumption is made on shards that are not involved in a transaction: they are free to form coalitions with other malicious actors in the system and try to subvert the guarantees of the protocol.
		
	\smallskip
	
	\noindent{\bf Liveness.} \scheme ensures the following liveness properties:
	\begin{itemize}
		\item If the protocol execution involves only honest-but-curious shards, then the execution terminates (all discard or all commit).
		
		\item If the protocol execution involves malicious shards, then the execution does not necessarily terminate. However, the shards can still engage in other protocol executions that do not update the state that is being locked.  
	\end{itemize}

	\smallskip
	\noindent{\bf Privacy} \scheme ensures that: 
	\begin{itemize}
		\item At the end of the interaction between the stakeholders and a (potentially-malicious) shard, the shard only learns: i) the stakeholders' long-term and ephemeral public keys; ii) the states that need to be updated and how; and iii) the public keys of the other shards involved. Stakeholders however do not disclose to the shard the content of the updates that other shards need to perform. 
				
		\item A transaction submitted to the ledger does not reveal (to network participants) any information about which stakeholders or shards are involved in the transaction. Also, the transaction does not leak any information about the expected updates. The transaction only reveals i) the number of shards in the transaction ii) and the ephemeral public key of the client who submitted it.
	\end{itemize}

	\subsection{Protocol Overview}\label{subsec:protocol-overview}
	In \scheme{} transaction lifecycle is divided into three phases, the \emph{pre-ordering phase}, the \emph{ordering phase}, and the \emph{post-ordering phase}. 
	
	\smallskip
	
	In the \emph{pre-ordering phase}, the stakeholders of the states to be updated interact with the shards to assemble the transaction. 
	First, stakeholders set up the terms of the transaction (i.e. expected state updates). For each state that needs to be modified, the respective stakeholders contact the shard that is responsible for the state with an \emph{update request}. If the request meets the corresponding stakeholder policy, then the shard acknowledges correct reception with a signed hash. Finally, the stakeholders exchange their update requests and the corresponding shard responses, verifying that the update requests match what has been agreed-upon and that the shard signatures in the received responses are valid. The stakeholders create the transaction by signing the hashes of the update requests and bundling the results, and relay the transaction to the client assigned to submit it to the ledger.
	\smallskip
	
	In the \emph{ordering phase}, the transaction is submitted to the ledger that orders it w.r.t. to the rest of the transactions submitted therein.

\smallskip
       The \emph{post-ordering phase}, refers to the transaction processing that takes place when the shards receive the ordered transaction from the ledger. At this point, each shard checks whether the transaction refers to a previously-acknowledged update request, and, if that is the case, the shard engages in an atomic commit protocol with all the other shards involved in the transaction. At the end of this phase, all shards either update their state or discard the transaction. Shards ignore any ledger transaction that they do not recognize. 
 	\section{Simplified MSPT protocol}
	\label{sec:black-box}
	
	We start our protocol description with a simplified version abstracting shards as centralized, single nodes (we generalize this to decentralized shards in Section~\ref{sec:generalize}).
	
	Consider  two stakeholders collaborating to create  a transaction. The stakeholders' states are held by two shards. For better understanding, we explain the protocol with the concrete example of a coin exchange, that can be easily generalized to other use-cases.

	\noindent\textbf{Coin exchange scenario.}
	Alice and Bob, the two stakeholders, wish to exchange \(100 \ \texttt{aCoins}\) for \(100 \  \texttt{bCoins}\). The \texttt{aCoins} accounts are managed by shard $S_A$, whereas the \texttt{bCoins} accounts are managed by shard $S_B$. We assume Alice and Bob to be in possession of long term signature key-pairs $\langle \sf{pk}_{A, B}, \sf{sk}_{A, B}\rangle$ they use to authenticate themselves to the system. To ensure fairness, the two transfers underlying the exchange must be executed atomically. 
	
	\subsection{Pre-Ordering Phase}
	\noindent\textbf{Stakeholder negotiation.}
	\label{sec:neg}
	The negotiation phase is the first part of the transaction pre-ordering phase. Alice and Bob prepare the update request for shard \(S_A\) to transfer  \(100\) \texttt{aCoins} to Bob. This negotiation happens ``off-chain'', i.e., exclusively  between the two stakeholders, without the involvement of any shard or ledger. For the purpose of the exchange transaction, Alice and Bob each generate, an ephemeral signing key-pair,  $\langle {\sf epk}_{x}, {\sf esk}_{x}\rangle, x = A, B$.  The resulting \emph{update request} $\sf req_{S_A}$ has the following form 
\[
{\sf req}_{S_A} = \langle 
{\sf id}, {\sf n}, {\sf p}, {\sf dep}, {\sf pk}, {\sf epk}, \Sigma %
\rangle.
\]

	\noindent where
	\begin{itemize}
	\item ${\sf id}$ is a unique request identifier, that is used to link all update requests that pertain to the same atomic transaction. 
	\item The random nonce ${\sf n}$ is used by the shard to randomize its response. The combination of the request identifier and the nonce allows shard \(S_A\) to deduplicate requests and prevent replay attacks.
	\item ${\sf p}$ is an opaque request payload. It is meant to contain a description of the way the stakeholders wish to modify the state. Its content is not prescribed by the protocol, leaving the freedom to each shard to define a suitable format.
	\item ${\sf dep}$ represents the dependency set of the proposed state changes, identifying the shards that \(S_A\) needs to collaborate with to make the transaction atomic. By setting ${\sf dep} = \{S_{B}\}$ as the dependency set to the message, Alice and Bob tell $S_{A}$ that it needs to coordinate with $S_{B}$ to commit the update. Namely, $S_{A}$ commits its changes if and only if $S_{B}$ also commits its changes. %
	\item The stakeholder public keys (${\sf pk} = [{\sf pk}_{A}, {\sf pk}_{B}]$) and signatures ([$\Sigma = [\sigma_{{\sf pk}_{A}}, \sigma_{{\sf pk}_{B}}$]) are used by the shard to validate the stakeholder policy of the state representing Alice and Bob's accounts. 
	\item The stakeholders' ephemeral public keys ($ {\sf ePK} = [{\sf epk}_{A}, {\sf epk}_{B}]$) will be used in the post-ordering phase to verify that the correct stakeholders signed the resulting transaction, without using their long-term public keys which would otherwise reveal their identity to the network.
	\end{itemize}
	
	\noindent Eventually, Alice (or Bob) passes ${\sf req}_{S_A}$ to \(S_A\).

	\noindent\textbf{Shard response.}
	\label{sec:resp}
	When $S_{A}$ receives the update request, it first checks if it has received the same request in the past and if so, it rejects it. Then $S_{A}$ unpacks the request and validates the long-term public keys and respective signatures against the stakeholder policy for the state that is involved in the transaction (in this case the state representing Alice and Bob's accounts). If this validation is unsuccessful, \(S_A\) simply responds with an error message. Otherwise, it hashes ${\sf req}_{S_A}$ using a secure hash function, into ${\sf hreq}_{S_A}$ and stores the key-value pair  $\langle  {\sf hreq}_{S_A}, {\sf  req}_{S_A} \rangle$ in its local transient store.  \(S_A\) concludes the response phase by sending a response $\sf{resp}_{S_A}$ to the origin of ${\sf req}_{S_A}$, containing \(S_A\)'s signature $\sigma_{S_A}^r$ on ${\sf hreq}_{S_A}$. 
	When the response is received and shared among the stakeholders, they can verify that the response (i) comes from the correct shard and (ii) belongs to the previous update request.
	
	\noindent\textbf{Multiple requests per transaction.}
	The steps outlined above are repeated for all update requests that make up the transaction. Consequently, Alice and Bob go through the same steps to update their accounts of \texttt{bCoins} on shard $S_B$. In particular, (i) they reuse the same request identifier ${\sf id}$ and (ii) add $\{S_{A}\}$ as the dependency set to the request to shard  \(S_B\). Note that even though \(S_A\) and \(S_B\) learn about each other's involvement, they do not learn anything about the content of each other's update request. 
	In practice, multiple update requests to different shards can be handled concurrently. %
	In fact, shards contact each other to coordinate only during validation.
	
	\noindent\textbf{Transaction creation.}
	\label{sec:txcreation}
	When the responses of shards \(S_A\) and \(S_B\) are received, Alice and Bob build the transaction as follows. They bundle the hashes from the shard responses into a single transaction and sign the result relative to the ephemeral public keys. This additional round of signatures is necessary to prevent a Byzantine stakeholder, for example Alice, from mixing and matching shard responses and unilaterally submitting the resulting transaction. Now, both Alice and Bob are sure that either the coin exchange is executed correctly or both retain their original coins.
	Alice (or Bob) then passes the transaction to a client to be submitted. The client accordingly signs the transaction and submits it to the ledger.

	\subsection{Ordering}
	\label{sec:floword}
	When the ledger receives a new transaction from the client, it verifies the client’s signature and puts the transaction into a new block, which is then disseminated to the whole network, i.e. received by all shards. 
	
	\subsection{Post-Ordering Phase}
	\noindent\textbf{Validation \& commitment.}
	\label{sec:val}	
	Both shards $S_{A}$ and $S_{B}$ scan every transaction in each new block they receive. They look up the hashes in their transient local store. Shards ignore hashes that they do not recognize, as that means that the shard did not generate the hash. At some point, $S_{A}$ and $S_B$ will recognize the hashes of their respective update requests. For each hash a shard recognizes, it verifies if the transaction carries valid signatures with respect to the stakeholders' ephemeral public keys. If that were not the case, then the transaction could have been created without the knowledge of one of the stakeholders. Note that shards only verify the set of signatures corresponding to their known update requests. They can ignore other signatures potentially belonging to parts of the transaction directed at other shards.
	
	After signature verification, shards update their state according to the instructions described in the payload ${\sf p}$ of the request. MSPT allows a transaction to contain multiple update requests to the same shard. Hence, shards simulate all their update requests in the order they show up in the transaction to guarantee consistent execution. The combined execution results form their local belief about the validity of the transaction. With that, they engage with the shards in the request's dependency set in a privacy preserving atomic commit protocol, whereby the unique request ID serves as the session identifier. For readability purposes, we dedicate Section \ref{sec:PPC} to describe our atomic commit protocol.  
	At the end of this protocol, the shards either commit the result of their respective update requests to their local state databases or discard them.
	
	\subsection{Privacy Discussion}

	\scheme consists of an initial phase of private siloed communication channels transitioning into a phase of public broadcast. We now point out potential information leaks due to this transition and how we prevent them.
	\begin{itemize}
        \item The identity of the submitting client is leaked to all ledger nodes and all shard nodes. However, if the ledger supports anonymous authentication (e.g. Idemix in the case of Hyperledger Fabric), stakeholders can submit transactions directly while hiding their identities. No additional trust assumptions are needed in this case. If the ledger does not support anonymous authentication, then the stakeholders can resort to a {\em pool of trusted clients} that act as a sort of a mixing network. A stakeholder will choose at random one client in the pool and delegate that client to submit its transaction.
        \item The identity of the shards involved in a transaction is only known to stakeholders who initiated it; shards that are involved in a transaction have a partial view about other involved shards, namely, they only know their immediate dependency set. Other entities learn nothing about that; this includes all ledger nodes, all other stakeholders and shards.
        \item The identity of the stakeholders that are involved in a transaction is only known to other involved stakeholders as part of the pre-order negotiation. A subset of involved stakeholders is revealed to the shards in the request messages. Despite the fact that stakeholders have to sign the request that is publicly broadcast and appended to the ledger, the reader will notice that the signature is generated with an ephemeral key pair, whose linkage to the long-term identity of the stakeholder is only revealed privately to each shard. As a consequence, other entities learn nothing about this.
        \item The request payload is only known by stakeholders and shards involved in the request that contains it. The rest of the network learns nothing about it, owing to the fact that only a hash is disseminated and the hash contains a random nonce,  which will guarantee sufficient entropy to avoid brute-force attacks over the hash value.
        \item Transactions reveal the number of involved shards and stakeholders. An approach to mitigate this issue is to add bogus update requests and signatures. 
	\end{itemize}
	
	\section{Privacy-preserving Atomic Commit (PPAC) Protocol}
	\label{sec:PPC}

	By assuming that shards (as a whole) are honest and the network is partially synchronous, we can use a blocking commit protocol instead of employing a full BFT consensus. However, to preserve the privacy of shards participating in the protocol, we cannot rely on all-to-all communication or single coordinators like the widely-used two-phase atomic commit (2PC)~\cite{gray1978notes, Bernstein1987}.
	
	Recall that the dependency sets included in the update requests identify which shards must atomically commit. Therefore, the connection between a shard and the shards in its dependency set form a natural channel of communication. Consequently, for our privacy-preserving atomic commit protocol (PPAC) shards exchange information along the edges of the dependency graph spanned by the union of all dependency sets in the transaction. Therefore, each shard only communicates with its direct neighbors in the graph (i.e. shards in its dependency set). 
	Before proceeding, let us (i) define what a dependency set 
	is and (ii) how a dependency set impact the dependency graph.
	A dependency set ${\sf dep}$ is a set consisting of shard identifiers. 
	A shard identifier $S_i$ can appear in the set only once and in the
	following form: $S_i$, $S_i+$, or $S_i-$. 
	Then, a dependency graph $G_{{\sf dep}}$ has a vertex for each shard identifier
	in the union of the dependency sets (modulo the sign).
	In addition, if a request for shard $S_j$ contains in its dependency set: 
	\begin{itemize}
		\item $S_i$ or $S_i+$: Then the dependency graph contains a
		direct edge from $S_j$ to $S_i$.
		\item $S_i-$: Then the dependency graph contains a
		direct edge from $S_i$ to $S_j$.
	\end{itemize}	
	The $(-,+)$ signs help each shard to perform proper access control.
	Indeed, if a request for shard $S_j$ contains in its dependency set: 
	\begin{itemize}
		\item $S_i+$: Then, $S_j$ should contact $S_i$.
		\item $S_i-$: Then, $S_j$ expects to be contacted by $S_i$.
		\item $S_i$: $S_j$ should contact $S_i$ and should also 
		expect to be contacted by $S_i$.
	\end{itemize}	
	Therefore, the shape of the graphs depends on the transaction's dependency structure, examples are illustrated in Figure~\ref{fig:graphs}.
	
	Note that instrumenting shard identifiers with signs forces the graph to be directed. This controls the flow of information and, for example, satisfies the scenario where a shard needs to provide information to another shard, but must not receive any information back from the recipient.
	Using the dependency graph as the communication overlay in this way guarantees that every shard learns all the necessary information to decide whether to commit or discard a transaction, and no more than that information. First, note that transitive dependencies form paths in the graph. Accordingly, we consider a shard to be connected to another shard if and only if it is (transitively) dependent on that shard. To ensure that all the necessary data is propagated along the graph edges the protocol runs in multiple rounds. 
	In Appendix~\ref{sec:bounds}, we prove that the protocol decides after at most $n-1$ rounds, where $n$ is the number of hashes in the transaction, and after as little as one round. To give some intuition for this result, when $n$ shards are involved in the protocol, information takes at most $n-1$ rounds to reach from one shard to any other, assuming that it can only travel one edge per round. This bound can be calculated by all shards individually and is known to them at the start of the protocol. 
	
	In the following, we describe our privacy-preserving commit protocol with an example run illustrated in Figure~\ref{fig:PPC}, while pseudo-code is in Appendix~\ref{sec:pseudocode}.
	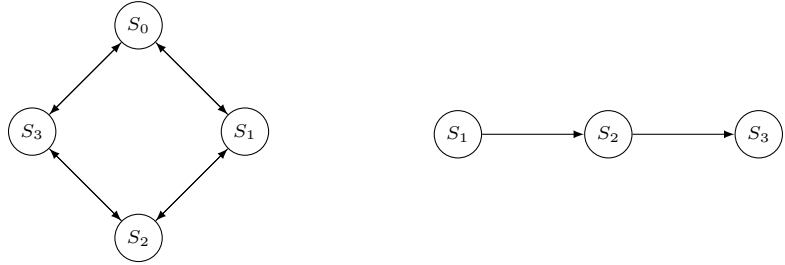
\begin{figure*}[t]
		\scriptsize
		\centering
		\begin{subfigure}[t]{0.48\textwidth}
			\centering
			\begin{tikzpicture}[node distance = 2cm]
				\node[draw, circle](s0) {$S_0$};
				\node[draw, circle, below right of=s0](s1) {$S_1$};		
				\node[draw, circle, below left of=s1](s2) {$S_2$};
				\node[draw, circle, below left of=s0](s3) {$S_3$};			
				\draw[->, -latex] (s0) -- (s1);
				\draw[->, -latex] (s1) -- (s0);
				\draw[->, -latex] (s1) -- (s2);
				\draw[->, -latex] (s2) -- (s1);
				\draw[->, -latex] (s2) -- (s3);
				\draw[->, -latex] (s3) -- (s2);
				\draw[->, -latex] (s3) -- (s0);
				\draw[->, -latex] (s0) -- (s3);
			\end{tikzpicture}
			\caption{Mutual dependency: 
			This graph is induced by the following dependency 
			sets: ${\sf dep}_0={S_1,S_3}$, ${\sf dep}_1={S_0,S_2}$, ${\sf dep}_3={S_1,S_3}$, 
			${\sf dep}_3={S_0,S_2}$}
			\label{fig:graph2}
		\end{subfigure}
		\hspace{0.2cm}
		\begin{subfigure}[t]{0.48\textwidth}
			\centering
			\begin{tikzpicture}[node distance = 2cm]
				\node[draw, circle](s2) {$S_2$};		
				\node[draw, circle, left of=s2](s1) {$S_1$};
				\node[draw, circle, right of=s2](s3) {$S_3$};
				\draw[->, -latex] (s1) -- (s2);
				\draw[->, -latex] (s2) -- (s3);
				\node[draw, circle, below right of=s2, opacity=0]{s4}; %
				\node[draw, circle, above right of=s2, opacity=0]{s4}; %
			\end{tikzpicture}
			
			\caption{Transitively directed dependency: 
			This graph is induced by the following dependency 
			sets: ${\sf dep}_1={S_2+}$, ${\sf dep}_2={S_1-,S_3+}$, ${\sf dep}_3={S_1-}$,}
			\label{fig:graphs1}
		\end{subfigure}
		\caption{Examples of dependency graphs}
		\label{fig:graphs}
	\end{figure*}
		 
	\noindent\textbf{Start.} By design, a transaction can only be formed after all involved shards have seen their corresponding update requests. Thus, at the start of the atomic commit protocol for a specific transaction, each involved shard is able to validate the update request and decide on a starting state value \texttt{\{tentative commit, discard\}}.

	\noindent\textbf{Protocol round.} Each involved shard queries the shards in its dependency set for their current states. More precisely, only shards marked with a $+$ sign or no sign are contacted. A shard responds to a query if and only if the invoking shard appears in the dependency set with a $-$ sign or no sign. If a shard receives a \texttt{discard} message, then it must update its own value to \texttt{discard}.
	
	\noindent\textbf{End.} Once a shard has completed $n - 1$ rounds it can finalize the transaction according to its final value, since it has received the cumulative state of all its (transitive) dependencies. Note that this means that it now processes exclusively its own update request. The only thing the shard learns during the protocol is which action its dependencies will choose. Specifically, it learns nothing about the update requests of the other shards.

	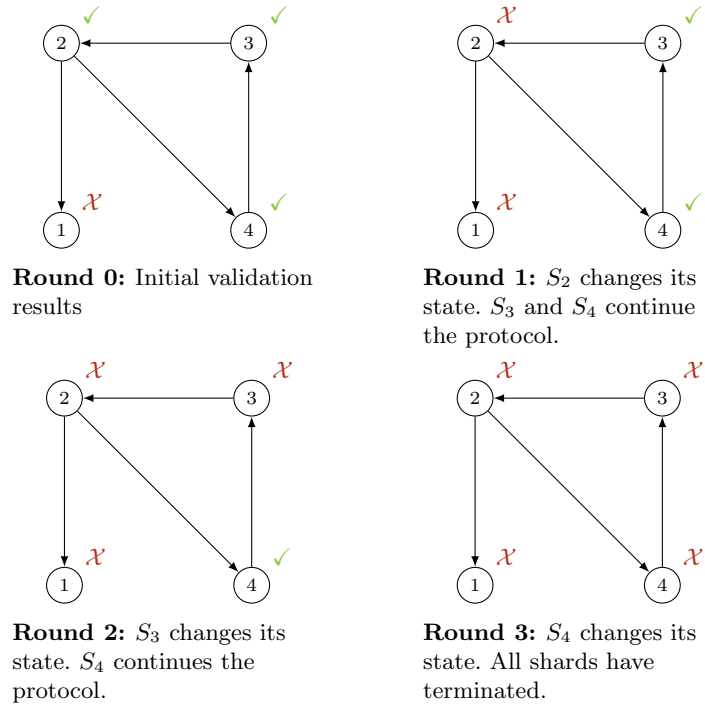
\begin{figure*}[t]
	\scriptsize
	\centering
	\captionsetup[subfigure]{labelfont=bf,textfont=normalfont,singlelinecheck=off,justification=raggedright,labelformat=simple, labelsep=colon}
	\renewcommand{\thesubfigure}{Round \arabic{subfigure}}

	\begin{subfigure}[t]{0.35\textwidth}
		\centering
\begin{tikzpicture}
[
    sh/.style={draw, node distance=2cm, circle},
    an/.style={node distance=0mm, font=\footnotesize}
]

	\node[sh] (b) {2};
	\node[sh, right=of b] (c) {3};
	\node[sh, below= of c] (d) {4};
	\node[sh, below=of b] (a) {1};
	
	\draw[-latex](b) -- (a);
	\draw[-latex](c) -- (b);
	\draw[-latex](d) -- (c);
	\draw[-latex](b) -- (d);
	
	\node[an,above right=of a, text=BrickRed]{$\mathcal{X}$};
	\node[an,above right=of b, text=LimeGreen]{\checkmark};
	\node[an,above right=of c, text=LimeGreen]{\checkmark};
	\node[an,above right=of d, text=LimeGreen]{\checkmark};

	\end{tikzpicture} 		\setcounter{subfigure}{-1}
		\caption{Initial validation results}
	\end{subfigure}
	\hspace{1cm}
	\begin{subfigure}[t]{0.35\textwidth}
		\centering
\begin{tikzpicture}
[
    sh/.style={draw, node distance=2cm, circle},
    an/.style={node distance=0mm, font=\footnotesize}
]

	\node[sh] (b) {2};
	\node[sh, right=of b] (c) {3};
	\node[sh, below= of c] (d) {4};
	\node[sh, below=of b] (a) {1};
	
	\draw[-latex](b) -- (a);
	\draw[-latex](c) -- (b);
	\draw[-latex](d) -- (c);
	\draw[-latex](b) -- (d);
	
	\node[an,above right=of a, text=BrickRed]{$\mathcal{X}$};
	\node[an,above right=of b, text=BrickRed]{$\mathcal{X}$};
	\node[an,above right=of c, text=LimeGreen]{\checkmark};
	\node[an,above right=of d, text=LimeGreen]{\checkmark};

	\end{tikzpicture} 		\caption{$S_2$ changes its state. $S_3$ and $S_4$ continue \\ the protocol.}			
	\end{subfigure}

	\begin{subfigure}[t]{0.35\textwidth}
		\centering
\begin{tikzpicture}
[
    sh/.style={draw, node distance=2cm, circle},
    an/.style={node distance=0mm, font=\footnotesize}
]

	\node[sh] (b) {2};
	\node[sh, right=of b] (c) {3};
	\node[sh, below= of c] (d) {4};
	\node[sh, below=of b] (a) {1};
	
	\draw[-latex](b) -- (a);
	\draw[-latex](c) -- (b);
	\draw[-latex](d) -- (c);
	\draw[-latex](b) -- (d);
	
	\node[an,above right=of a, text=BrickRed]{$\mathcal{X}$};
	\node[an,above right=of b, text=BrickRed]{$\mathcal{X}$};
	\node[an,above right=of c, text=BrickRed]{$\mathcal{X}$};
	\node[an,above right=of d, text=LimeGreen]{\checkmark};

	\end{tikzpicture} 		\caption{$S_3$ changes its state. $S_4$ continues the protocol.}
	\end{subfigure}
	\hspace{1cm}
	\begin{subfigure}[t]{0.35\textwidth}
		\centering
\begin{tikzpicture}
[
    sh/.style={draw, node distance=2cm, circle},
    an/.style={node distance=0mm, font=\footnotesize}
]

	\node[sh] (b) {2};
	\node[sh, right=of b] (c) {3};
	\node[sh, below= of c] (d) {4};
	\node[sh, below=of b] (a) {1};
	
	\draw[-latex](b) -- (a);
	\draw[-latex](c) -- (b);
	\draw[-latex](d) -- (c);
	\draw[-latex](b) -- (d);
	
	\node[an,above right=of a, text=BrickRed]{$\mathcal{X}$};
	\node[an,above right=of b, text=BrickRed]{$\mathcal{X}$};
	\node[an,above right=of c, text=BrickRed]{$\mathcal{X}$};
	\node[an,above right=of d, text=BrickRed]{$\mathcal{X}$};

	\end{tikzpicture} 		\caption{$S_4$ changes its state. All shards have terminated.}			
	\end{subfigure}

	\caption{Example of an execution of our \texttt{PPAC} algorithm. 
	The dependency graph is induced by the following dependency sets:
	${\sf dep}_1=\{S_2-\}$, ${\sf dep}_2=\{S_3-,S_4+\}$, ${\sf dep}_3=\{S_2+,S_4-\}$, and 
	${\sf dep}_4=\{S_2-,S_3+\}$.
	Each shard is only aware of its direct neighbors. Dependencies are marked by arrows.}
	\label{fig:PPC}
\end{figure*}

	\noindent\textbf{Synchrony.} While we have described the rounds as if all shards would go through them synchronously, in reality shards do not directly coordinate and start the protocol as soon as they find a transaction that involves them. To still guarantee deterministic termination of the protocol, shards can keep track of their round numbers and the round numbers of their neighbors by counting state queries. Any shard sends exactly one query to its dependencies per round. Based on these counters, we can introduce some coupling between neighbors: a shard will not respond to a query unless it itself has received all responses from its dependencies in the previous round. The querying shard has no choice but to wait until it receives a response. This way, neighbors in the dependency graph go through the protocol in lockstep and always see the correct state from their dependencies.

	\noindent\textbf{Optimization.} Shards might be able to finalize their state early. First, a \texttt{discard} can never be changed back to \texttt{commit}, so any discarding shard, while still responding to queries, can discard the transaction immediately. Second, shards can come to different upper bounds based on their dependency sets, as discussed in Appendix~\ref{sec:bounds}. To take advantage of this, we introduce a third value \texttt{finalized commit}, which shards can respond with when they have already finalized. A shard that in turn only receives \texttt{finalized commit} responses from all its dependencies can finalize its own state, because it can be sure that the states of dependent shards will never change again. This way the early commit or discard can be cascaded along the dependency graph.
	
	Performance of the protocol can be further enhanced by shards analyzing dependencies between the transactions they are involved in~\cite{Gorenflo2019a}. If the update requests would alter distinct local state the shard can start the atomic commit protocol for those transactions in parallel.
	
	\noindent\textbf{Transaction atomicity.} While we have stated as a goal to achieve atomic commit for transactions as a whole, in reality in MSPT atomicity is defined by the dependencies between requests. If two requests are not dependent on each other then there is no need for them to atomically commit. Moreover, a transaction can theoretically include requests that span multiple disconnected dependency graphs. Our protocol guarantees that requests \emph{connected by dependencies} atomically commit. Disconnected requests effectively form multiple transactions bundled in the same message, so this does not break the security goals stated in Section~\ref{sec:spec}.

	\section{Experiments}
	\label{sec:tngpoc}

	\noindent\textbf{Implementation.} We implemented a proof-of-concept of simplified MSPT on Hyperledger Fabric~v1.4. Our MSPT-Fabric retains all of Fabric's features, with the added capability of processing cross-shard transactions created by our protocol. According to Section~\ref{sec:neg}, the creation of update requests is handled completely off-chain. A request is then sent to a shard, which is implemented as a Fabric endorser with a thin wrapper. In this context, each peer/endorser of the Fabric network represents a distinct shard (the extension to shards comprising more than one node is described in Section~\ref{sec:generalize}). The wrapper unpacks the message and translates it into a Fabric transaction proposal to forward it to the endorser. The response is directed back through the wrapper, where it is stored together with the update request in a transient store. Only a hash is returned to the client according to Section~\ref{sec:resp}. The ledger is implemented over the ordering service of Fabric, which stays unchanged as it does not unpack the contents of a transaction. Lastly, we added a new flag \texttt{PRIVATE$\_$TRANSACTION} to the grpc messages Fabric uses for communication. This way, a peer can correctly deserialize a private transaction during validation, look up the update request in the transient store and engage in the \texttt{PPAC} from Section~\ref{sec:PPC}.
	
	We implement the same toy example as in Section~\ref{sec:black-box} with a total of up to five shards.
	
	\noindent\textbf{Experimental setup.} We use nine local servers: one client, three orderers forming a Raft cluster (MSPT-Fabric ledger) and five endorsing peers (shards). Each of these processes is spawned in a Docker container on its own server equipped with two Intel\textsuperscript{\textregistered} Xeon\textsuperscript{\textregistered}~CPU~E5-2620~v2 processors at 2.10~GHz, for a total of 24~hardware threads and 64~GB of RAM.		
	We compare three degrees of privacy features and their performance trade-off: Full MSPT-Fabric, MSPT-Fabric with \texttt{2PC} instead of \texttt{PPAC} and Fabric without privacy features. Transactions in this setup only modify non-overlapping state, so both atomic commit protocols can run parallelized. Up to five shards collaborate on a single transaction with transitive directed dependencies between them as shown in Figure~\ref{fig:graphs1}. Hence, \texttt{PPAC} will terminate in zero (single shard) to four rounds (five shards), while \texttt{2PC} terminates in zero (single shard) or one round (other configurations). Each experimental run sends 10,000 transactions to the network.

\begin{figure*}[t]
	\centering
	\scalebox{0.55}{
		\begin{subfigure}[t]{0.6\textwidth}
			
\begin{tikzpicture}[scale=1]
		\pgfplotsset{cycle list/Set1}
        \begin{axis}[        
        legend style={at={(1,0.95)}, anchor=south east, font=\small},
        legend cell align=left,
        ymin=0,
        ymax=390,
        width= \linewidth,
        ylabel= Throughput in tx/s,
        xlabel = {\# shards per tx},
        xtick=data,
        cycle list={[indices of colormap={1,0,2 of Set1}]}
        ]       
        
        \pgfplotstableread[col sep=comma]{tp.csv}\data
        
        \addplot+[thick, mark=square,mark options={solid},error bars/.cd, y dir = both,y explicit] table[x=x, y=fab_mean, y error=fab_std]from \data;
        \addplot+[thick, mark=triangle,mark options={solid},dashed, error bars/.cd, y dir = both,y explicit,error bar style={solid}] table[x=x, y=ptx_mean, y error=ptx_std]\data;
        \addplot+[thick, mark=x,mark options={solid}, dotted, error bars/.cd, y dir = both,y explicit,error bar style={solid}] table[x=x, y=ptx_tp_mean, y error=ptx_tp_std]from \data;
        \legend{ Fabric v1.4,MSPT-Fabric + \texttt{PPAC}, MSPT-Fabric + \texttt{2PC}};
        \end{axis}
\end{tikzpicture}

			\caption{sequential update requests}
			\label{fig:tpa}
		\end{subfigure}
		\hfill
		\begin{subfigure}[t]{0.6\textwidth}
			
\begin{tikzpicture}[scale=1]
		\pgfplotsset{cycle list/Set1}
        \begin{axis}[        
        legend style={at={(1,0.95)}, anchor=south east, font=\small},
        legend cell align=left,
        ymin=0,
        ymax=3900,
        width= \linewidth,
        ylabel= Throughput in tx/s,
        xlabel = {\# shards per tx},
        xtick=data,
        cycle list={[indices of colormap={1,0,2 of Set1}]}
        ]       
        
        \pgfplotstableread[col sep=comma]{tp.csv}\data
        
        \addplot+[thick, mark=square,mark options={solid},error bars/.cd, y dir = both,y explicit] table[x=x, y=fab_pe_mean, y error=fab_pe_std]from \data;
        \addplot+[thick,mark=triangle,mark options={solid},dashed, error bars/.cd, y dir = both,y explicit,error bar style={solid}] table[x=x, y=ptx_pe_mean, y error=ptx_pe_std]\data;
        \addplot+[thick,mark=x,mark options={solid}, dotted, error bars/.cd, y dir = both,y explicit,error bar style={solid}] table[x=x, y=ptx_pe_tp_mean, y error=ptx_pe_tp_std]from \data;
        \legend{ Fabric v1.4,MSPT-Fabric + \texttt{PPAC}, MSPT-Fabric + \texttt{2PC}};
        \end{axis}
\end{tikzpicture}
 			\vspace{-0.4cm}
			\caption{pre-processed update requests}
			\label{fig:tpb}
		\end{subfigure}
		
		\begin{subfigure}[t]{0.6\textwidth}
\begin{tikzpicture}[scale=1]
	\begin{axis}[        
	legend style={at={(1,0.95)}, anchor=south east, font=\small},
	legend cell align=left,
	ymin=0,
	ymax=29,
	width= \linewidth,
	ylabel= Latency in ms,
	xlabel = {\# shards per transaction},
	xtick=data,
	cycle list/Set1,
	cycle list={[indices of colormap={0,2,3 of Set1}]},
	]       
	
	\pgfplotstableread[col sep=comma]{lt.csv}\data
	
	\addplot+[thick,mark=triangle,mark options={solid},dashed, error bars/.cd, y dir = both,y explicit,error bar style={solid}] table[x=x, y=ptx_mean, y error=ptx_std]\data;

	\addplot+[thick,mark=square,mark options={solid},dashed, error bars/.cd, y dir = both,y explicit,error bar style={solid}] table[x=x, y=tpc_mean, y error=tpc_std]\data;
	\legend{ \texttt{PPAC},  \texttt{2PC}};
	\end{axis}
\end{tikzpicture}
 			\caption{atomic commit protocol only}
			\label{fig:lat}
		\end{subfigure}
	}
	
	\caption{Impact of private transaction creation and \texttt{PPAC} on Fabric's performance.}
	\label{fig:tp}
\end{figure*}
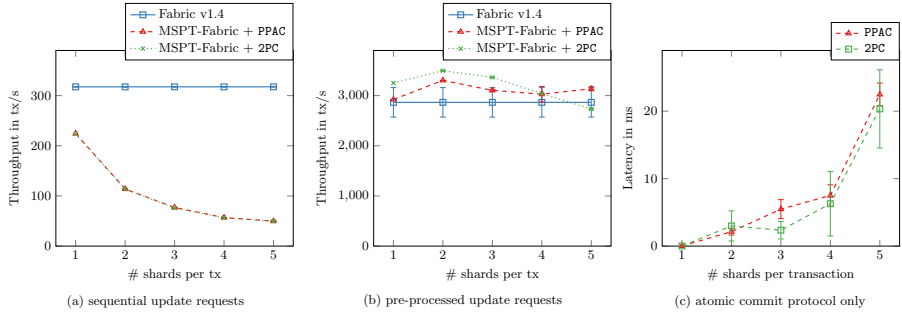	
	
	\noindent\textbf{Throughput.}		
	Experiments are conducted in two modes: First, we measure the full end-to end throughput starting with sequential update requests to the shards and ending with transaction commits. Second, stakeholders send all update requests to all shards and create all transactions before the measurements start and we only measure the throughput of the validation and commitment step.
	
	In the end-to-end experiment, we find that the update request processing results in a performance hit of about 30\% (Figure~\ref{fig:tpa}) compared to regular Fabric. For more than one shard, the throughput follows an expected $\frac{1}{n}$ relation because of the sequential processing of $n$ update requests before submitting the transaction, where $n$ MSPT update requests translate to $n$ regular Fabric transactions. The choice of atomic commit protocols is irrelevant for end-to-end throughput. Focusing purely on the validation and commitment throughput (Figure~\ref{fig:tpb}), all tested systems show very similar performance. Regular Fabric seems to even be outperformed by MSPT-Fabric. However, this is due to MSPT-Fabric skipping parts of the validation. Remodeling these parts would have meant a complete rewrite of existing Fabric code, which was beyond the scope of the proof-of-concept. More importantly, \texttt{2PC} outperforms \texttt{PPAC} only slightly, without offering any privacy features.
		
	\noindent\textbf{Latency.} Our final experiment (Figure~\ref{fig:lat}) investigates the average latency of a single execution of the atomic commit protocols. Surprisingly, \texttt{2PC} just barely manages to come out ahead and shows similar degradation of performance for an increasing number of shards. 
	
	\noindent\textbf{Implication.} Comparing the results between the two throughput experiments, it is clear that the processing of update requests is the main bottleneck of the protocol. Processing private update requests imposes a constant overhead of about 30\% on the transaction throughput. However, this can be amortized with an implementation of MSPT-Fabric that allows concurrent instead of only sequential update request processing.	
	With the help of parallel execution of the atomic commit protocols, the validation and commitment step can easily keep up with the pre-processing. 
	
	Since both  \texttt{PPAC} and \texttt{2PC} show surprisingly similar results in both throughput and latency experiments despite diverging numbers of communication rounds between shards, we conclude that at least in a setup with local servers the performance of the commit protocols is mostly governed by Fabric's block dissemination: statistically, the more participants in a transaction, the larger the average gap between the first and last shard receiving the same block. While \texttt{2PC} must wait for all shards to communicate their validation results to terminate, \texttt{PPAC} can already make progress if neighbors in the dependency graph are ready. This weaker synchrony requirement can help offset the fewer rounds of communication of \texttt{2PC}. 
	
	In conclusion, in our experiments, MSPT with \texttt{PPAC} is almost five times faster than the reported performance of private contracts on JP Morgan Chase's Quorum blockchain~\cite{Baliga2018a}, which has a less capable privacy model than MSPT.
	
	\section{Generalizing MSPT}
	\label{sec:generalize}	
	\label{sec:anybc}
	\label{sec:fullMSPT}
	
So far we described MSPT as an extension to Hyperledger Fabric assuming centralized shards. 
In this section, we first discuss how other blockchain systems could implement MSPT. Then we describe the necessary changes to run MSPT with decentralized shards. 
	
	\noindent\textbf{Arbitrary permissioned blockchains.}
	MSPT can be configured to work with both order-execute  (OX) or  execute-order (XO) blockchains \cite{Vukolic2017}. Recall that shards only send a hash back to the stakeholders. By deliberately decoupling request execution from the response to the stakeholders in that manner the shard is free to execute the state transition any time between receiving the request and starting the cross-shard commit protocol.
	
	For use in OX blockchains (e.g., Enterprise Ethereum), shards wait with the execution until they see the corresponding hash in the ledger. However, they need a rollback strategy in case they must discard after the protocol completes. This is a similar restriction to rolling back if the transaction provides not enough gas.
	
	In contrast, XO blockchains like Fabric can simulate the result as soon as they receive the request. After an invalid simulation, they can immediately report back an error so this request will never enter a transaction. On success, they store the result until they spot the corresponding nonce in a transaction. Then they only need to validate that the state transition has not become stale in the meantime before starting the cross-shard commit protocol.

	\noindent\textbf{Decentralization of shards.}
	For the generalization of shards we first look at shard interfaces. First we consider interaction of shards with stakeholders. Recall that we model fault-prone shard nodes as crash-faulty.
	
	When a stakeholder sends an update request to a shard in the centralized case, we simply replace this by a stakeholder reliably broadcasting \cite{CachinGR11} an update request to a decentralized shard. Reliable broadcast guarantees that eventually, if a single correct shard node delivers an update request, all correct shard nodes will. 
	
	When a shard is about to send a response to a stakeholder, we replace this step by having all shard nodes send responses to a stakeholder. This, together with the reliable broadcast step of an update request guarantees liveness --- even if a single shard node is correct (non-faulty), a stakeholder is guaranteed to receive a response.
	
	As for the validation \& commitment step, apart from PPAC invocation, this step is local to a shard and involves only deterministic steps. Note that due to reliable broadcast of update request, all shard nodes have locally the update request, once it is delivered on the ledger. Therefore, we generalize this step by having every shard node execute validation \& commitment step independently. 
	
	Finally, to generalize PPAC (Sec.~\ref{sec:PPC}), we let every shard node run PPAC independently, and we replace a step in which a shard $S_A$ queries shard $S_B$, by having every shard node in $S_A$ query every shard node in $S_B$ and awaiting until it receives the first response from some shard node in $S_B$. Since shard nodes are assumed to be crash-faulty (this assumption itself guarantees safety) and if at least one shard node is correct per shard, then this ensures liveness.

	Some of these generalization steps could be solved by putting a consensus (total order (TO) broadcast) protocol in the critical path of the shard processing. However, this can have a considerable impact on performance and is strictly speaking not necessary. However, use of efficient TO broadcast protocol, such as \cite{Stathakopoulou2019} remains an option.
	
\bibstyle{abbrv}
\printbibliography
\appendix
\section{Bounds on communication rounds}
\label{sec:bounds}

For the purposes of this discussion, we use the definition of a round of communication from Section~\ref{sec:PPC}. This presents a logical batching rather than a temporal one, since shards are only loosely synchronized through their counters. Yet, under the assumption of a partially synchronous system there is a bounded delay between two shards executing the same round of communication. In the following, we examine how many rounds of communication are necessary for shards to atomically commit a given transaction.

\begin{lemma}
	\label{lem:edge}
	Each round of communication can only relay information along the edges of the dependency graph. Any information can only travel the length of one edge per round.
\end{lemma}
\begin{proof}
	Assume that a given shard $S_i$ can communicate with a shard $S_j$, where $S_i$ and $S_j$ are not directly connected by the dependency graph. By necessity, $S_i$ would need to learn about $S_j$, even though $S_j$ is not in $S_i$'s dependency set and vice versa. This is explicitly forbidden.
\end{proof}

\begin{lemma}
	\label{lem:paths}
	For any two given shards $S_i$ and $S_j$ there exists a finite or countably infinite set of paths between them. 	
\end{lemma}
\begin{proof}
	The number of edges of the dependency graph is finite, otherwise the underlying transaction would be infinitely large. Therefore we can enumerate them all. For the purposes of this proof we choose a counting system with a base at least as large as the number of edges. This way, each edge is labeled by a unique symbol.
	
	Starting from $S_i$ we create all possible paths by traversing all connecting edges and writing down the symbols of the edges. Then, for each vertex we landed on, we repeat this process, appending the new edge number to the previous sequence. The process for a given path ends when $S_j$ is reached. Since the dependency graph can have cycles, sequences can be arbitrarily long. But, all paths from $S_i$ to $S_j$ must be finite by definition. Therefore, we can sort the sequences belonging to those paths in ascending order and count them.
\end{proof}
\begin{lemma}
	\label{lem:sp}
	For any two given connected shards $S_i$ and $S_j$ it takes exactly $l$ rounds of communication to relay information from one to the other, where $l$ is the length of the shortest path between the shards. If there is no path, then $l = \infty$.
\end{lemma}
\begin{proof}
	As by Lemma~\ref{lem:paths}, the paths between two shards are countable, so we can order all paths between $S_i$ and $S_j$ by their length. Now, we let all shards relay all their information to all neighbors in every round of communication. According to Lemma~\ref{lem:edge}, the information can make progress along each path one edge at a time. We can track this by subtracting 1 from the length of each path. After $l$ rounds the shortest path will have 0 length, indicating that $S_j$ received the information.
	
	There cannot be a shorter path that is not on the list, because by definition we listed all possible paths between $S_i$ and $S_j$.
	
	If there is no path in the ordered list, then this search will never stop, i.e. $l=\infty$. 
\end{proof}

\begin{theorem}
	Let \texttt{PPAC} be an algorithm that ensures atomic commitment of a transaction $tx$ while keeping the global set of participants secret. Then, \texttt{PPAC} must go through at least $l^*$ rounds of communication, where $l^*$ is the length of the longest $\emph{shortest path}$ between any two connected shards of the dependency graph for $tx$.
\end{theorem}

\begin{proof}
	For a transaction to commit correctly, every shard must commit or discard the transaction atomically together with all other dependent shards. By necessity, a shard must learn all validation results of the shards it is transitively dependent on because a single \texttt{discard} signal would change its own belief. According to Lemma~\ref{lem:sp}, for every shard $S_i$, every shortest path to another connected shard $S_j$ must be traversed. If two shards are not connected, they do not require to learn each other's belief, so for all relevant paths is $l<\infty$. The transaction is not finalized before the last one of them reaches its destination. Therefore, the longest length of all shortest paths forms the lower bound on finalization.
\end{proof}

Note that this is only a lower bound on global transaction finalization. Single shards are able to finalize their beliefs in the following four cases:
\begin{itemize}
	\item It discarded the transaction in its own validation.
	\item It is not dependent on any other shard.
	\item It received a \texttt{discard} signal from one of its collaborating shards.
	\item All of its collaborating shards signal that they finalized their belief.
\end{itemize}

A finalized shard must broadcast its belief to its neighbors before dropping out of the commit protocol.

If shards do not encounter any of these cases, they might not stop after $l^*$ rounds, because they are not able to construct the full dependency graph. Therefore, $l^*$ is not known to the shards.

\begin{lemma}
	\label{lem:degree}
	For a given shard $S_i$, it's longest \emph{shortest path} to another connected shard $S_j$ has at most length $n-\delta^+(S_i)$, where $n$ is the number of nonces in the transaction and $\delta^+(S_i)$ is the out-degree of vertex $S_i$.
\end{lemma}	
\begin{proof}
	For $S_i$, $\delta^+(S_i)$ shards are connected by a shortest path of length 1. This leaves at most $n-\delta^+(S_i) -1$ unknown shards\footnote{Some of the nonces could be associated with the same shard or be invalid.}. In the worst case they are positioned as a chain with  $n-\delta^+(S_i) -2$ edges connecting them. This chain of shards must be connected to $S_i$ via one of its direct neighbors. Let $S_j$ be the last vertex in the chain. Then, the shortest path from $S_i$ to $S_j$ is 
	\begin{align*}
		l^*&=\underbrace{1}_{\text{\makebox[30pt][c]{\tiny neighbor}}} + \underbrace{1}_{\text{\makebox[30pt][c]{\tiny connect chain}}}+ (n-\delta^+(S_i)-2) \\
		&=n-\delta^+(S_i).	
	\end{align*}
\end{proof}
\begin{theorem}
	\label{the:2}
	All shards finalize their belief after at most $n-\min\limits_i \delta^+(S_i)$ rounds of communication.
\end{theorem}
\begin{proof}
	Directly follows from Lemma~\ref{lem:degree}.
\end{proof}

	\section{Pseudocode}
	\label{sec:pseudocode}
\begin{minipage}[H]{\linewidth}
	\centering
	\begin{algorithm}[H]
	\caption{\texttt{PPAC}\label{alg:PPC}}
	\begin{algorithmic}	
		\Function{ValidateTx}{\emph{requestID, initialBelief, hashCount, deps}} 
		
		\Comment \emph{hashCount}: number of hashes in the transaction.
		\Comment \emph{deps}: shard dependencies with a $+$ or no sign.
		
		\State $b \gets initialBelief$
		\Comment{\emph{false = discard}, \emph{true = commit}}
		
		\State
		\If {$\text{\Call{Count}{\emph{deps}}} = 0\ \vee\ b = \emph{false}$}
		\State \Call{StoreBelief}{\emph{requestID}, $b$, \emph{isFinal}: true} 
		
		\Comment{save to internal storage beyond the sketch of this algorithm}
		\State \Return $b$
		\EndIf
		\State
		\For{$\text{\emph{round}}\gets 1,\ \text{\emph{responsesCount}} - \text{\Call{Count}{\emph{deps}}}$}
		\State \Call{StoreRound}{\emph{requestID}, \emph{round}} \Comment{save to internal storage}
		\State \Call{StoreBelief}{\emph{requestID}, $b$, \emph{isFinal}: false}
		\State\emph{allFinal} $\gets \text{true}$
		\State
		\ForAll{\emph{dep} \textbf{in} \emph{deps}}
		\State\Comment{remote call to other shard, caller made explicit for easier understanding}
		\State $b_{dep}$, \emph{isFinal} $\gets \text{\emph{dep}.\Call{PullBelief}{\emph{requestID}, thisShard}}$ 
		
		\State \emph{allFinal} $\gets$ \emph{allFinal} $\wedge$ \emph{isFinal}
		\State $b \gets b\ \wedge\ 	b_{dep}$
		\State	
		\If {\textbf{not} $b$}
		\State break all loops
		\EndIf
		\EndFor
		\State
		\If {\emph{allFinal}}		
		\State break loop
		\EndIf
		\EndFor
		\State 
		\State \Call{StoreBelief}{requestID, $b$, \emph{isFinal}: true}
		\State \Return $b$
		\EndFunction
		\State
		\State
		
	\end{algorithmic}
\end{algorithm}
\end{minipage}

\begin{algorithm}[H]
	\caption{\texttt{PullBelief} is called remotely by \texttt{PPAC} on a different shard.\label{alg:pull}}
	\begin{algorithmic}	
		
		\Function{PullBelief}{\emph{requestID}, \emph{shard$_{\text{caller}}$}} 
		
		\State\Comment function is called remotely from other shard, in practice \emph{shard$_{\text{caller}}$} would be obtained implicitly from caller address
		\State $i \gets \text{\Call{GetCallCount}{\emph{requestID}, $\emph{shard$_{\text{caller}}$}$}}$ 
		\Comment{read from internal storage}
		\State
		\While{{\text{\Call{GetRound}{\emph{requestID}}} $\leq i$ }}
		\State Wait
		\EndWhile 
		\State $b, \text{\emph{isFinal}} \gets$ \Call{GetBelief}{\emph{requestID}} 
		\Comment{read from internal storage}
		\State \Call{StoreCallCount}{\emph{requestID}, \emph{shard$_{\text{caller}}$},  $i + 1$}
		\Comment{save back to internal storage}
		\State \Call{RespondTo}{\emph{shard$_{\text{caller}}$}, $b$, \emph{isFinal}} 
		\Comment{send back current belief}
		
	\EndFunction

	\end{algorithmic}
\end{algorithm}
\end{document}